\documentclass[copyright,creativecommons]{eptcs}
\providecommand{\event}{GANDALF 2011} 

\usepackage[english]{babel}
\usepackage[T1]{fontenc}
\usepackage{amsmath,amssymb,latexsym, booktabs}
\usepackage{graphicx}
\usepackage{pifont}
\usepackage{newalg}

\providecommand{\urlalt}[2]{\href{#1}{#2}}
\providecommand{\doi}[1]{doi:\urlalt{http://dx.doi.org/#1}{#1}}

\def\altbox{\hspace{2mm}\nolinebreak\null\nolinebreak\hfill\Box}
\newenvironment{proof}{\noindent {\bf Proof:}}{$\altbox$\bigskip}

\newenvironment{tracex}{\ccspace{0.15}\noindent}{\ccspace{0.2}}

\newcommand{\ccspace}[1]{\vspace{#1cm}}

\newtheorem{theorem}{Theorem}[section]
\newtheorem{definition}[theorem]{Definition}

\newtheorem{proposition}[theorem]{Proposition}
\newtheorem{example}[theorem]{Example}

\def\name#1{\mbox{\sc #1}}
\def\sos#1#2{{\def\arraystretch{1.6}\begin{array}{c}#1\\\hline
#2\end{array}}}

\newcommand{\Nil}{{\sf nil}}
\newcommand{\ignore}[1]{}

\newcommand{\gor}{\;\big|\;}
\newcommand{\rec}{{\sf rec} \:}

\newcommand{\nar}[2]{\xrightarrow{#1}_{#2}}
\newcommand{\sar}[1]{\nar{#1}{}}

\newcommand{\ur}{{\cal U}}
\newcommand{\B}{{\mathbb A}} 
\newcommand{\Bt}{\B_{\tau}} 
\newcommand{\X}{{\cal X}}
\newcommand{\Pg}{\tilde{\PG}}
\newcommand{\PG}{{\mathbb P}}
\newcommand{\Sg}{\tilde{\SG}}
\newcommand{\SG}{{\mathbb S}}

\newcommand{\clean}{{\sf clean}}
\newcommand{\unlab}{{\sf unmark}}

\newcommand{\fase}{\texttt{FASE}}

\newcommand{\cwbnc}{\texttt{CWB-NC}}

\newcommand{\rop}{\triangleright}

\newcommand{\bool}{{\mathbb B}}

\newcommand{\kvalues}{{\mathbb K}}

\newcommand{\fa}{\mathit{f}}
\newcommand{\tr}{\mathit{t}}

\newcommand{\Bv}{{\sf B}}

\newcommand{\Cv}{{\sf C}}

\newcommand{\rtb}[1]{{\it b_{{\rm #1}}rt}}
\newcommand{\rfb}[1]{{\it b_{{\rm #1}}rf}}
\newcommand{\wtb}[1]{{\it b_{{\rm #1}}wt}}
\newcommand{\wfb}[1]{{\it b_{{\rm #1}}wf}}
\newcommand{\wtc}[1]{{\it c_{{\rm #1}}wt}}
\newcommand{\wfc}[1]{{\it c_{{\rm #1}}wf}}
\newcommand{\rtc}[1]{{\it c_{{\rm #1}}rt}}
\newcommand{\rfc}[1]{{\it c_{{\rm #1}}rf}}
\newcommand{\Kv}{{\sf K}}

\newcommand{\rk}[1]{{\it kr #1}}
\newcommand{\wk}[1]{{\it kw #1}}
\newcommand{\rc}[1]{{\it c_{#1}r}}
\newcommand{\wc}[1]{{\it c_{#1}w}}
\newcommand{\ruk}{\rk{1}}
\newcommand{\rdk}{\rk{2}}
\newcommand{\wuk}{\wk{1}}
\newcommand{\wdk}{\wk{2}}

\newcommand{\checky}{\ding{51}}
\newcommand{\checkn}{\ding{55}}

\newcommand{\vp}{{\sf PV}}

\newcommand{\req}{{\tt req}}

\newcommand{\cs}{{\sf cs}}
\newcommand{\dekker}{{\it Dekker}}
\newcommand{\dijkstra}{{\it Dijkstra}}
\newcommand{\knuth}{{\it Knuth}}
\newcommand{\peterson}{{\it Peterson}}
\newcommand{\lamport}{{\it Lamport}}

\newcommand{\BK}{{\sf BK}}

\newcommand{\dijkstraa}{{\sf Dijkstra}} 
\newcommand{\knutha}{{\sf Knuth}} \newcommand{\knuthb}{\knutha}
\newcommand{\petersona}{{\sf Peterson}} 

\newcommand{\lamporta}{{\sf Lamport}} 
\newcommand{\proc}{{\tt P}}

\newcommand{\ppath}[5]{#1 & #2 & #3 & #4 & #5}

\title{Automated Analysis of MUTEX Algorithms with $\fase$
\thanks{
This work was supported by the PRIN Project `Paco:Performability-Aware
Computing: Logics, Models, and Languages'.}}
\author{Federico Buti, Massimo Callisto De Donato, Flavio Corradini, Maria Rita Di Berardini
\institute{School of Science and Technology, University of Camerino}
\email{\{federico.buti, massimo.callisto, flavio.corradini,
mariarita.diberardini\}@unicam.it}
\and
Walter Vogler
\institute{Institut f\"ur Informatik, Universit\"at Augsburg}
\email{vogler@informatik.uni-Augsburg.de}}
\def\titlerunning{Automated Analysis of MUTEX Algorithms with $\fase$}
\def\authorrunning{F.~Buti et al.}
\begin{document}
\maketitle
 
\begin{abstract}
In this paper we study the liveness of several MUTEX solutions by
representing them as processes in PAFAS$_s$, a CCS-like process algebra
with a specific operator for modelling non-blocking reading behaviours.
Verification is carried out using the tool \fase, exploiting a
correspondence between violations of the liveness property and a special
kind of cycles (called {\em catastrophic cycles}) in some transition
system. We also compare our approach with others in the literature. The aim
of this paper is twofold: on the one hand, we want to demonstrate the
applicability of \fase\  to some concrete, meaningful examples; on the
other hand, we want to study the impact of introducing non-blocking
behaviours in modelling concurrent systems.
\end{abstract}

\section{Introduction}
MUTEX algorithms can exhibit an intricate behaviour and their correctness
can be hard to establish, because our intuitive notion of the program flow
can be misled by the fact that a shared variable may change from one
statement to the other, even if the process we are tracing does not modify
it. There are two kinds of properties to verify: the {\em safety property}
that two competing processes are never in their critical sections at the
same time, and the {\em liveness property} that a requesting process will
always enter its critical section. The first kind of property can be proven
fairly easily because only the static configuration of the system at any
time must be taken into account. The liveness property is much more
difficult to prove since it usually requires some fairness assumption.

In~\cite{CVJ02}, we have developed the process description language PAFAS,
a CCS-like~\cite{Mil89} process algebra  originally introduced as a tool
for evaluating the worst-case efficiency of asynchronous systems. Processes
are compared via a variant of the testing approach of De~Nicola and
Hennessy~\cite{DH84} where tests are test environments (or user behaviours)
together with a time bound. A process is embedded into the environment (via
parallel composition) and satisfies a (timed) test, if success is reached
before the time bound in {\em every} run of the composed system, i.e.\ even
in the worst case. This gives rise to a {\em faster-than} preorder relation
over processes that is naturally an {\it efficiency preorder}.
%
%
In~\cite{CV05} it has been shown that the test-based preorder
in~\cite{CVJ02} can equivalently be defined on the basis of a performance
function that gives the worst-case time needed to satisfy any test
environment. Whenever the above testing scenario is adapted to a setting
where tests belong to a very specific, but often occurring, class of {\em
request-response} user behaviours (processes serving these users receive
requests via an $in$-action and provide responses via an $out$-action) this
performance function is {\em asymptotically linear}. This provides us with
a quantitative measure of systems performance
that measures how fast the system under consideration responds to requests
from the environment. In~\cite{CV05} we have also shown how to determine
this performance measure for finite-state processes. This result only holds
for request-response processes (i.e. processes that can only perform $in$
and $out$ as visible actions) that pass certain sanity checks: they must
not produce more responses than requests, and they must allow requests and
provide responses in finite time. While the first requirement can easily
be read off from the transition system, violation of the latter one
is characterised as the existence of a special kind of cycles (called
{\em catastrophic cycles}) in a reduced transition system (we remind the reader to \cite{CV05} for the complete description of such a reduction).
Finally, a corresponding tool \fase\ that allows the automated
evaluation of systems performance function has been developed; see
\cite{BCCDV09} for a first informal account. 

The notion of timing in PAFAS is strongly related to (weak) {\em fairness
of actions} which requires that an action must be performed whenever it is
enabled continuously in a run. We have shown that each everlasting (or
non-Zeno) timed process execution is fair and vice versa, where fairness is
defined in an intuitive but complicated way in the spirit
of~\cite{CostaS84,CostaS87}. In fact, we have proven this correspondence
for fairness of actions and, with a modified notion of timing, for fairness
of components. These characterisations have been used in~\cite{CDV06} to
prove that Dekker's algorithm is live under the assumption of fairness of components
but not under the assumption of fairness of actions. This result can be
improved by means of suitable assumptions about the hardware, namely we
must assume that reading a value from a storage cell is non-blocking; to
model this we have introduced specific reading prefixes for PAFAS
in~\cite{CDV08tr}. 

Here, we add reading in the form of a read-set prefix $\{a_1, \ldots, a_n\}
\rop P$ (the new process description language is called PAFAS$_s$) which
behaves as $P$ but, like a variable or a more complex data structure, can
also be read with actions in the set $\{a_1, \ldots, a_n\}$. Since being
read does not change the state, each action $a_i$ ($i=1,\ldots,n$) can be
performed repeatedly until the execution of some ordinary action of $P$ .

A first key property of non-blocking actions is that they have a direct
impact on timed behaviour of concurrent systems (see the examples at the
end of Section~\ref{PAFAS-s}). They are also an important feature for
proving the liveness of MUTEX solutions under the assumption of weak
fairness of actions. Indeed, one result in~\cite{CDV08tr} shows that
Dekker's algorithm is live when assuming fairness of actions, provided we
regard as non-blocking the reading of a variable as well as its writing in
the case that the written value equals the current one. It had long been an
open problem how to achieve such a result in a process algebra (see
e.g.~\cite{Walker89}). In~\cite{CDV08tr} we have also discovered an
interesting connection between liveness of MUTEX algorithms and
catastrophic cycles; we have shown that violations of the liveness property
can be traced back to catastrophic cycles of a suitably modified process
(cf. Section~\ref{algos}). Even though \fase\ was originally developed for
automatically checking whether a process of (original) PAFAS has a
catastrophic cycle, it has been recently adapted to a setting with reading
actions. This has opened the way to check automatically the liveness
property for MUTEX algorithms.

In this paper we use \fase\ to study the liveness of four MUTEX
solutions--Peterson's, Lamport's, Dijkstra's and Knuth's algorithms (see
\cite{Walker89} and references therein)--under the assumption of
fairness of actions. Our aim is twofold: we want to show the applicability
of \fase\ to  concrete, meaningful examples, but also to stress the
impact of introducing non-blocking actions in PAFAS (and in general in
modelling concurrent systems). We prove that Peterson is live provided we
regard the reading of a variable as a non-blocking action. We also show
that the liveness of Dijkstra and Knuth cannot be ensured even if (as
in~\cite{CDV08tr}) we consider as non-blocking the reading of a variable
and its writing in the case the written value equals the current one. With
the same assumption on program variables, we finally prove that Lamport
(which is not symmetric) is live for just one of the two competing
processes, i.e.\ it is not live. 

To even more emphasize the role of non-blocking reading in proving
liveness property, we have implemented some ideas taken
from~\cite{Walker89}
that describe how fairness can be assumed in a CCS setting in order to
enable a proof of liveness. At the time of writing, these ideas could not
be expressed for the use of the Concurrency
Workbench~\cite{CleavelandPS89}, but this is now possible within newer
tools like the Concurrency Workbench of the New Century~\cite{CLS00}. A
comparison of the results provided by the two approaches shows that the
liveness of Dekker's and Peterson's algorithms strongly depends on the
liveness of the hardware. This is exactly the sort of consideration for
which non-blocking actions provide a formal treatment.

We proceed as follows: In Section~\ref{PAFAS-s} we recall PAFAS$_s$, its
timed operational semantics and the correspondence between fair traces and
everlasting timed computations. In Section~\ref{algos} we introduce the
four algorithms and provide our results. Finally, in
Section~\ref{sec:walker} we compare our approach with that
in~\cite{Walker89}.
\section{A process algebra for describing reading
behaviours}\label{PAFAS-s}
PAFAS~\cite{CVJ02} is a CCS-like process description language~\cite{Mil89}
(with a {\em TCSP}-like parallel composition), where actions are atomic and
instantaneous, but have associated an upper time bound (either $0$ or $1$,
for simplicity) as a maximal delay for their execution. As shown
in~\cite{CVJ02}, these upper time bounds can be used to evaluate
efficiency, but they do not influence functionality (which actions are
performed); so compared with CCS also PAFAS treats the full functionality
of asynchronous systems. In~\cite{CDV08tr}, PAFAS has been extended with a
new operator $\rop$ to represent non-blocking behaviour of processes.
Intuitively, $\{\alpha_1, \ldots, \alpha_n\}\rop P$ models a process like a
variable or a more complex data structure that behaves as $P$ but can
additionally be read with $\alpha_1,\ldots, \alpha_n$: since being read
does not change the state, each action $\alpha_i$ can be performed
repeatedly without blocking a synchronization partner as described below.
We use the following notation: $\B$ is an infinite set of \emph{visible
actions}; the additional action $\tau$ represents a internal activity,
unobservable for other components, and $\Bt=\B\cup\{\tau\}$. Elements of
$\B$ are denoted by $a,b,c,\dots$ and those of $\Bt$ by
$\alpha,\beta,\dots$. Actions in $\Bt$ can let time $1$ pass before their
execution, i.e.\ 1 is their maximal delay. After that time, they become
\emph{urgent} actions written $\underline{a}$ or $\underline{\tau}$; these
have maximal delay 0. The set of urgent actions is denoted by
$\underline{\B}_\tau = \{\underline{a} \:|\: a \in \B \}\cup \{
\underline{\tau} \}$ and is ranged over by
$\underline{\alpha},\underline{\beta},\dots$. Elements of $\Bt \cup
\underline{\B}_\tau$ are ranged over by $\mu$ and $\nu$. We also assume
that, for any $\alpha \in \Bt$, when time elapses $\underline{\underline{\alpha}} =
\underline{\alpha}$. $\X$ (ranged over by $x, y, z, \ldots$) is the set of
process variables, used for recursive definitions. $\Phi \: : \: \Bt
\rightarrow \Bt$ is a {\it general relabelling function} if the set $\{
\alpha \in \Bt \: | \: \emptyset \neq \Phi ^{-1}(\alpha) \neq \{ \alpha
\}\}$ is finite and $\Phi(\tau)= \tau$. Such a function can also be used to
define {\em hiding}: $P/A$, where the actions in $A$ are made internal, is
the same as $P[\Phi_A]$, where the relabelling function $\Phi_A$ is defined
by $\Phi_A(\alpha)= \tau$ if $\alpha \in A$ and $\Phi_A(\alpha)=\alpha$ if
$\alpha \notin A$. 

Below, initial processes are just processes of a standard process algebra
extended with $\rop$, while general processes are those reachable from the
initial ones according to the operational semantics.
%
The set $\Sg_1$ of {\em initial (timed) process terms} $P$ and $\Sg$ of
(general) {\em (timed) process terms} $Q$ are generated by:

\vspace{0.1cm}
\hspace{3cm}
$P::= \Nil \gor x \gor \alpha. P \gor \{\alpha_1, \dots ,\alpha_n\} \rop P
\gor P+P \gor P \:\|_A \: P \gor P[\Phi] \gor \rec x.P$

\hspace{3.2cm} 
$Q ::= P \gor \underline{\alpha} .P \gor \{\mu_1, \dots ,\mu_n\}
\triangleright Q \gor Q+Q \gor Q \:\|_A\: Q \gor Q[\Phi] $

\vspace{0.1cm}

where  $\Nil$ is a constant, $x \in \X$, $\alpha \in \Bt$,  $\Phi$ is a
general relabelling function and $A\subseteq \B$ possibly infinite;
$\{\alpha_1, \ldots,\alpha_n\}$   and $\{\mu_1,\ldots, \mu_n\}$ are (finite
and nonempty) subsets of  $\Bt$ and $\Bt \cup \underline{\B}_\tau$, resp.
We assume that the latter kind of read-sets can only contain a copy (either
lazy or urgent) of each action $\alpha$, i.e.\ $\{\mu_1,\ldots, \mu_n\}$
cannot contain both $\alpha$ {\em and} $\underline{\alpha}$ for any $\alpha
\in \Bt$. By the operational semantics, terms not satisfying this property
are not reachable from initial ones anyway.  
%
%
A process term is {\em closed} if every variable $x$ is bound by the
corresponding $\rec x$-operator; the set of closed timed process terms in
$\Pg$ and $\Pg_1$, simply called {\em processes} and {\em initial
processes} resp., is denoted by $\PG$ and $\PG_1$ resp.

$\Nil$ is the Nil-process: it cannot perform any action but can let time
pass without limits. $\alpha.P$ and $\underline{\alpha}.P$ is
action-prefixing known from CCS. Process $\alpha.P$ performs $\alpha$
within time 1; i.e.\ it can perform $\alpha$ immediately and evolve to $P$
(as usual in CCS), or let one time unit pass and become
$\underline{\alpha}.P$. In this latter case, $\alpha$ cannot be further
delayed (i.e.\ it must occur or be deactivated) unless
$\underline{\alpha}.P$ has to wait for a synchronisation on $\alpha \neq
\tau$. Our processes are {\em patient}: as a stand-alone process
$\underline{a}.P$ has no reason to wait, but as a component of a larger
system, e.g.\ $\underline{a}.P \,\|_{\{a\}}\, a.\Nil$, it can wait for a
synchronisation on $a$; this can take up to time 1 since component $a.\Nil$
can idle so long. $\{\mu_1,\ldots,\mu_n\} \rop Q$ can perform actions from
$\{\mu_1,\ldots,\mu_n\}$ without changing state (including urgencies and,
hence, the syntax of the term itself), and the actions of $Q$ in the same
way as $Q$, i.e.\ the read-set is removed after such an action. $Q_1 + Q_2$
is a non-deterministic choice between two conflicting processes $Q_1$ and
$Q_2$. \ $Q_1$ and $Q_2$ run in parallel in $Q_1\|_A Q_2$ and have to
synchronize on all actions from $A$. $P[\Phi]$ behaves as $P$ but with
actions changed according to  $\Phi$. $\rec x. P$ models a recursive
definition; we often use equations to define recursive processes. 
\paragraph{Functional and temporal behaviour of PAFAS$_s$ processes.}
\label{secTimedSem}
We first introduce the transitional semantics describing the functional
behaviour of PAFAS$_s$ processes, i.e.\ which actions they can perform.
\begin{definition}\label{PACTSOS}\rm({\em functional operational
semantics})
Let $Q\in \Sg$ and $\alpha \in \Bt$. The SOS-rules defining the transition
relation  $\nar{\alpha}{}\subseteq(\Sg \times \Sg)$ (the {\em action
transitions}) are given in Table~\ref{Behaviour}\footnote{We do here
without $\clean$ and $\unlab$, used e.g.\ in~\cite{CDV06} to get
a closer relationship between states of untimed fair runs and timed
non-Zeno runs. They do not change the behaviour (up to an injective
bisimulation) and would complicate the setting.}. As usual, we write
$Q\nar{\alpha}{}Q'$ if $(Q,Q')\in\nar{\alpha}{}$ and $Q\nar{\alpha}{}$ if
there exists a $Q'\in \Sg$ such that $(Q,Q')\in\nar{\alpha}{}$. Similar
conventions will apply later on. We also define the set of the {\em
activated} or enabled actions to be the set of all $\alpha$ such that
$Q \nar{\alpha}{}$.
\end{definition}

\vspace{-0.8cm}
\begin{table}[tbh]
\small
\[
\begin{array}{c}
\name{Pref$_s$}\; \sos{\mu\in\{\alpha, \underline{\alpha}\}}
{\mu.P \nar{\alpha}{}P}\qquad 
\name{Sum$_s$}\; \sos{Q_1\nar{\alpha}{}Q'}{Q_1+Q_2
\nar{\alpha}{} Q'}  \qquad 
\name{Read$_{s1}$}\; \sos{ \mu_i \in \{\alpha, \underline{\alpha}\}}
{\{\mu_1, \ldots, \mu_n\} \rop Q \nar{\alpha}{} \{\mu_1, \ldots,
\mu_n\} \rop Q} \\
\name{Read$_{s2}$}\; \sos{Q \nar{\alpha}{} Q'} {\{\mu_1, \ldots, \mu_n\}
\rop Q \nar{\alpha}{} Q'}  \quad 
\name{Par$_{s1}$} \sos{\alpha\notin A,\;Q_1\nar{\alpha}{}Q'_1}
{Q_1\|_A Q_2 \nar{\alpha}{} Q'_1\|_A Q_2}\quad
\name{Par$_{s2}$} \sos{\alpha\in A,\;Q_1\nar{\alpha}{}Q'_1,\;
Q_2\nar{\alpha}{}Q'_2}{Q_1\|_A Q_2 \nar{\alpha}{} Q'_1\|_A Q'_2}\\
\name{Rel$_s$} \sos{Q\nar{\alpha}{}Q'}
{Q[\Phi]\nar{\Phi(\alpha)}{}Q'[\Phi]}\quad
\name{Rec$_s$} \sos{Q \{\rec x.Q/x\} \nar{\alpha}{} Q'}
{\rec x.Q \nar{\alpha}{} Q'}\\
\end{array}\]
\caption{Functional behaviour of PAFAS$_s$ processes}
\label{Behaviour}
\end{table}

Rules in Table~\ref{Behaviour} are quite standard. Timing can be
disregarded in \name{Pref$_s$}: when an action is performed, one cannot see
whether it was urgent or not, and thus $\underline{\alpha}.P \nar{\alpha}{}
P$; furthermore, component $\alpha.P$ has to act {\em within} time 1, i.e.\
it can also act immediately, giving $\alpha.P \nar{\alpha}{}P$. Rules
\name{Read$_{s1}$} and \name{Read$_{s2}$} say that $\{\mu_1,\ldots, \mu_n\}
\rop Q$ can either repeatedly perform one of its non-blocking actions or
evolve as $Q$. Other rules are as expected; symmetric rules have been
omitted.
Actually, the above SOS-rules describe reading in a sensible way only under
some syntactic restrictions, cf.~\cite{CDV08tr}. All the example processes
we consider here meet these restrictions.

We now define the refusal traces of a term $Q \in \Sg$. Intuitively a
refusal trace records, along a computation, which actions process $Q$ can
perform ($Q \nar{\alpha}{} Q'$, $\alpha \in \Bt$) and which actions $Q$ can
refuse to perform when time elapses ($Q\nar{X}{r} Q'$, $X \subseteq \B$).
$Q \nar{X}{r} Q'$ is called a (partial) {\it time-step}. The actions listed
in $X$ are not urgent; hence $Q$ is justified in not performing them, but
performing a time step instead. This time step is partial because it can
occur only in contexts that can refuse the actions not in $X$. If $X = \B$
then $Q$ is fully justified in performing this time-step; i.e., $Q$ can
perform it independently of the environment. In such a case, we say that
$Q$ performs a {\it 1-step} written $Q \nar{1}{}Q'$; moreover we often
write $\underline{Q}$ (the urgent version of $Q$) instead of $Q'$. To
provide the reader with a better intuition we observe that any $Q$ can
perform a 1-step whenever it can refuse to perform, because not urgent, all
its activated actions. In the next definition, $\ur(\{\mu_1, \ldots,
\mu_n\}) =  \{\alpha \,|\, \mu_i = \underline{\alpha} \mbox{ for some } i
\in [1,n]\}$ is the set of urgent actions in $\{\mu_1,\ldots,\mu_n\}$. 

\begin{definition}\label{PAFASRT}\rm ({\em refusal transitional semantics})
The SOS-rules in Table~\ref{rt-semantics} define $\nar{X}{r} \subseteq (\Sg
\times \Sg)$ where $X \subseteq \B$.
\end{definition}
\begin{table}[tbh]
\small
\[
\begin{array}{c}
\name{Nil$_{t}$}\; \sos{}{\Nil \nar{X}{} \Nil}\quad
\name{Pref$_{t1}$}\; \sos{}{\alpha.P \nar{X}{r}
\underline{\alpha}.P}\quad
\name{Pref$_{t2}$}\; \sos{\alpha \notin X \cup \{\tau\}}
{\underline{\alpha}.P \nar{X}{r} \underline{\alpha}.P} \quad
\name{Sum$_{t}$}\; \sos{Q_i\nar{X}{r}Q'_i \mbox{ for }
i=1,2}{Q_1+Q_2 \nar{X}{r} Q'_1 + Q'_2} \\
\name{Read$_{t}$}\; \sos{Q \nar{X}{r} Q', \; \ur(\{\mu_1, \ldots,
\mu_n\} ) \cap (X \cup \{\tau\}) = \emptyset}
{\{\mu_1, \ldots, \mu_n\} \rop Q \nar{X}{r}
\underline{\{\mu_1, \ldots, \mu_n\}} \rop Q'} \qquad
\name{Rel$_{t}$} \sos{Q\nar{\Phi^{-1}(X \cup \{\tau\})\backslash
\{\tau\}}{r}Q'} {Q[\Phi]\nar{X}{r}Q'[\Phi]}\\
\name{Par$_{t}$} \sos{Q_i \nar{X}{r} Q'_i \mbox{ for } i=1,2, X
\subseteq (A \cap (X_1 \cup X_2)) \cup ((X_1 \cap X_2)\backslash A)}
{Q_1\|_A Q_2 \nar{X}{r} Q'_1\|_A Q'_2} \qquad
\name{Rec$_t$} \sos{Q\{\rec x.Q/x\} \nar{X}{r} Q'}{ \rec x.Q \nar{X}{r}
Q'}\\
\end{array}\]
\caption{Refusal transitional semantics of PAFAS$_s$ processes}
\label{rt-semantics}
\end{table}
Rule \name{Pref$_{t1}$} says that a process $\alpha.P$ can let time pass
and can refuse to perform any action, while rule \name{Pref$_{t2}$} says
that a process $\underline{\alpha}.P$, can let time pass but action
$\alpha$ cannot be refused. Process $\underline{\tau}.P$ cannot let time
pass and cannot refuse any action; in any context, $\underline{\tau}.P$ has
to perform $\tau$ before time can pass further. Rule~\name{Par$_{t}$}
defines which actions a parallel composition can refuse during a time-step.
The intuition is that $Q_1 \|_A Q_2$ can refuse an action $\alpha$ if
either $\alpha \notin A$ ($Q_1$ and $Q_2$ can do $\alpha$ independently)
and both $Q_1$ and $Q_2$ can refuse $\alpha$, or $\alpha \in A$ ($Q_1$ and
$Q_2$ are forced to synchronise on $\alpha$) and at least one of them can
refuse the action, i.e.\ can delay it. Thus, an action in a parallel
composition is urgent (cannot be further delayed) only when all
synchronising \lq local\rq\ actions are urgent. Rule \name{Read$_{t}$}
says that $\{\mu_1, \ldots, \mu_n\} \rop Q$ can refuse the same actions as
$Q$ provided these are not urgent in $\{\mu_1, \ldots, \mu_n\}$;
moreover, as for the action-prefixing,  process $\{\mu_1, \ldots, \mu_n\}
\rop Q$ cannot let time pass and cannot refuse any action,  whenever one
of the urgent actions in $\{\mu_1, \ldots, \mu_n\}$ is a $\tau$. Other
rules are as expected. Again symmetric rules have been omitted.

In \cite{CVJ02}, it is shown that inclusion of refusal traces characterises
a testing-based faster-than relation that compares processes w.r.t. their
worst-case efficiency. In this sense, e.g.\ $P = \{a\} \rop b.\Nil$ is
faster than the functionally equivalent $P'  =\rec x.\, a.x +b.\Nil$, since
only the latter has the refusal traces $1a(1a)^{\ast}$.  After $1a$, $P'$
returns to itself (recursion unfolding creates fresh $a$ and $b$);
intuitively, $b$ is disabled during the occurrence of $a$, so $a$ and also
$b$ can be delayed again. In contrast, after a 1-step and any number of
$a$'s, $P$ turns into $\{\underline{a}\} \rop \underline{b}.\Nil$ and no
further 1-step is possible; read actions do not block or delay other
activities, they make processes faster. If $a$ models the reading of a
value stored by $P$ or $P'$ and two parallel processes want to read it,
this should take at most time 1 in a setting with non-blocking reads. And
indeed, whereas $P' \,\|_{\{a\}}\, (a.\Nil \,\|_{\emptyset}\, a.\Nil)$ has
the refusal trace $1a1a$, this behaviour is not possible for $P
\,\|_{\{a\}}\, (a.\Nil\,\|_{\emptyset}\, a.\Nil)$ since, when performing
$1a$, this evolves into e.g.\ $\{\underline{a}\} \rop \underline{b}.\Nil
\,\|_{\{a\}}\, (\Nil \,\|_{\emptyset}\, \underline{a}.\Nil)$, and
then 1 is not possible.

Another application of refusal traces is the modelling of {\it weak
fairness of actions}. Weak fairness requires that an action must be
performed whenever continuously enabled in a run. Thus, a run from $P$
above with infinitely many $a$'s is not fair; the read action does not
block $b$ or change the state, so the same $b$ is  always enabled but never
performed. In contrast, if $P'$ performs $a$, a fresh $b$ is created; in
conformance to~\cite{CostaS84}, a run from $P'$ with infinitely many $a$'s
is fair. In~\cite{CDV08tr}, generalising \cite{CDV06}, fair traces for
PAFAS$_s$ are first defined in an intuitive, but very complex fashion in
the spirit of~\cite{CostaS84,CostaS87} and then characterised: {\em they
are the sequences of visible actions occurring in transition sequences with
infinitely many 1-steps}. 
Due to lack of space, we cannot properly formulate this as a theorem, but
take it as a {\bf definition} of \emph{fair traces} instead. With this,
infinitely many $a$'s are a fair trace of $P'$ since it can repeat $1a$
indefinitely, but the fair traces of finite-state $P$ are those that end
with $b$. We use this definition of fair traces to study liveness property
of MUTEX solutions we consider in the next section. 

For request-response processes the transition system (built according to
Def.~\ref{PACTSOS} and~\ref{PAFASRT}) must be reduced as described
in~\cite{CV05}; a cycle in the resulting system is catastrophic if it
contains (at least) one time step but no $in$- or $out$-transition.

\section{Liveness property of MUTEX algorithms}\label{algos}
In this section we use the approach of~\cite{CDV08tr} to study the liveness
of four different MUTEX solutions: Peterson's, Lamport's, Dijkstra's and 
Knuth's algorithm. We first translate the algorithms into PAFAS$_s$ and
then use \fase\ to automatically decide whether each of them is live or
not. Negative results are discussed by means of counterexamples, i.e.\ fair
violating traces which are built from catastrophic cycles detected with
\fase. The results of this section are collected in
Table~\ref{tab:comparing2}.

\paragraph{Peterson's algorithm}\label{sec:peterson}
There are two processes $\proc_1$ and $\proc_2$, two Boolean-valued
variables $b_1$ and $b_2$, whose initial value is false, and a variable
$k$, which takes values in $\{1,2\}$ and whose initial value is arbitrary.
The $b_i$ variables are ``request'' variables and $k$ is a ``turn''
variable: $b_i$ is true if $\proc_i$ is requesting entry to its critical
section and $k$ is $i$ if it is $\proc_i$'s turn to enter its critical
section. Only $\proc_i$ writes $b_i$, but both processes read it. Process 
$\proc_i$ (with $i=1,2$) is described as follows;  $j$ is the index of the
other process:

\vspace{0.2cm}
\small
\begin{algorithm}{Peterson}{}
\begin{WHILE}{$true$}
\langle \mbox{non-critical section} \rangle; \\
b_i \= true; \quad k \= j;\\
{\bf while} \; b_j \mbox{ and } k = j \; {\bf do \; skip};\\
\langle \mbox{critical section} \rangle; \\
b_i \= false;
\end{WHILE}
\end{algorithm}
\normalsize
In our translation of the algorithm into PAFAS$_s$, we use essentially
the same coding as Walker in~\cite{Walker89}. Each variable is represented
as a family of processes. For example, the process $\Bv_1(\fa)$ denotes the
variable $b_1$ with value false. The {\em sort} of $\Bv_1(\fa)$ (i.e.\ the
set of actions it can ever perform) is $\{\rfb{1},\rtb{1}, \wfb{1},
\wtb{1}\}$ 
Unlike~\cite{Walker89}, we model the actions that correspond to the reading
of a variable (e.g.\ $\rtb{1}$ and $\rtb{2}$) as non-blocking. Below, we
let $\bool=\{\fa, \tr\}$ and $\kvalues=\{1,2\}$.
\begin{definition}\label{pet-algo1}\rm({\sl Peterson's algorithm})
Let $i\in\{1,2\}$. Program variables are represented as follows:

\hspace{0.5cm}
$
\begin{array}{l l}
\Bv_i(\fa) = \{\rfb{\it i}\} \rop (\wfb{{\it i}}.\Bv_i(\fa) + \wtb{\it
i}.\Bv_i(\tr)) \qquad &
\Kv(1) =  \{\ruk\} \rop (\wuk.K(1) + \wdk .\Kv(2))\\
\Bv_i(\tr) = \{\rtb{\it i}\} \rop (\wtb{{\it i}}.\Bv_i(\tr) + \wfb{\it
i}. \Bv_i(\fa)) &
\Kv(2) =  \{\rdk\}\rop (\wuk.K(1) + \wdk .\Kv(2))
\end{array}
$

\vspace{0.1cm}

\noindent
Given $b_1, b_2 \in \bool$, $k\in\kvalues$, we define $\vp(b_1, b_2, k)
= (\Bv_1(b_1) \,\|_{\emptyset}\, \Bv_2(b_2)) \,\|_{\emptyset}\, \Kv(k)$.
Processes $\proc_1$ and $\proc_2$ are represented by the following
PAFAS$_s$ processes: the actions $\req_i$ and $\cs_i$ indicate the request
to enter and the execution of the critical section by the process
$\proc_i$.

\hspace{3cm}
$\begin{array}{ll}
P_1 = \req_1.\wtb{1}.\wk{2}.P_{11} + \tau.P_{1} & 
P_2 = \req_2.\wtb{2}.\wk{1}.P_{21} + \tau.P_{2} \\
P_{11} = \rfb{2}.P_{13} + \rtb{2}.P_{12} &
P_{21} = \rfb{1}.P_{23} + \rtb{1}.P_{22}\\
P_{12} = \rk{2}.P_{11} + \rk{1}.P_{13} & 
P_{22} = \rk{1}.P_{21} + \rk{2}.P_{23}\\
P_{13} = \cs_1.\wfb{1}.P_{1} & 
P_{23} = \cs_2.\wfb{2}.P_{2}
\end{array}$

\noindent
Since no process should be forced to request by the fairness assumption,
$P_i$ has the alternative of an internal move, i.e.\ staying in its
non-critical section.
Peterson's algorithm is defined to be the PAFAS$_s$ process $\petersona
= ((P_1 \:\|_{\emptyset}\: P_2) \:\|_B\: \vp(\fa, \fa, 1))/B$; here (and in
the following) $B$ is the set of all actions except $\req_i$ and $\cs_i$
($i=1,2$). A MUTEX algorithm like Peterson's satisfies {\em liveness} if,
in every fair trace, each $\req_i$ is eventually followed by the respective
$\cs_i$.
\end{definition}

We now show how to modify the process \petersona\  such that it is live 
under the assumption of fairness of actions iff the modified process, that
we call $\petersona_{io}$, does not have catastrophic cycles. Observe that 
\fase\ only accepts request-response behaviours (having only $in$ and $out$
as visible actions) as input and, hence, it cannot be applied directly.
Moreover, \petersona\ can perform a 1-step followed by the two internal
actions of $P_1$ and $P_2$ (see Def.~\ref{pet-algo1}) giving a catastrophic
cycle which is not relevant for the liveness property. So, we modify
\petersona\ as follows: we first change $\req_1$ and $\cs_1$ into $\tau$'s
and $\req_2$ and $\cs_2$ into $in$ and $out$, resp.; we finally  delete the
$\tau$ summand of $P_2$. As in~\cite{CDV08tr} (see Theorem 8.2\footnote{The
proof of Theorem 8.2 we provide in~\cite{CDV08tr} is partly independent
from the specific algorithm we were analysing, i.e.\ Dekker's algorithm,
and it can be easily adapted to all the algorithms we consider in this
paper. From now on, we freely use the correspondence between liveness and
catastrophic cycles without explicitly proving it. In the following, if $P$
is a PAFAS$_s$ process that models a given MUTEX solution, we write
$P_{io}$ to denote the process we obtain by changing $P$ as \petersona.}),
we can prove that $\petersona_{io}$ does not have catastrophic cycles iff
each request from process $\proc_2$ will eventually be satisfied along fair
traces, i.e.\ iff \petersona\ is live for  process $\proc_2$ under the
assumption of fairness of actions. The liveness of \petersona\ follows by
the symmetry of the algorithm. In case of non-symmetric algorithms, as
e.g.\ Lamport, the liveness for processes $\proc_1$ and $\proc_2$  must be
proven separately. Since \fase\ has shown that \petersona$_{io}$ does not
have catastrophic cycles, our first result is:

\begin{proposition}\label{prop:pet-live}
$\petersona$ is live under the assumption of fairness of actions.
\end{proposition}
We now consider $\petersona'$, a slightly different specification of
Peterson where all actions -- including the reading of program variables --
are ordinary actions. E.g., in this version, we define $\Bv_i(\fa) =
\rfb{\it i}.\Bv_i(\fa) + \wfb{{\it i}}.\Bv_i(\fa) +
\wtb{{\it i}}. \Bv_i(\tr) $. Then,  $\petersona'$ can be defined as in
Def.~\ref{pet-algo1}.

\begin{proposition}\label{prop:pet-non-live} 
$\petersona'$  is not live under the assumption of fairness of actions.
\end{proposition}
\begin{proof}
\fase\  shows that $\petersona'_{io}$ has  catastrophic cycles as,
e.g., those in the next examples.
\end{proof}

The following example shows a timed computation along which both processes
$\proc_1$ and $\proc_2$ get stuck after a request. To ease understanding,
we leave the actions on program variables visible, i.e.\ we consider a
timed computation of $P = (P_1 \,\|_{\emptyset}\, P_2)
\,\|_B\,\vp(\fa,\fa,1)$. Indeed, by the operational semantics, we know that
$P$ behaves as $\petersona'$ as long as we rename actions in $B$ with
$\tau$. We will proceed in this fashion later on in this section.
Furthermore, we write $\vp(\underline{\tr},\underline{\tr}, 1)$ and
$\vp(\tr,\underline{\tr}, 1)$ to abbreviate $(\underline{\Bv_1 (\tr)}
\,\|_{\emptyset} \, \underline{\Bv_2(\tr)} ) \,\|_{\emptyset} \, \Kv(1)$
and $(\Bv_1 (\tr) \,\|_{\emptyset} \,\underline{\Bv_2(\tr)} )
\,\|_{\emptyset} \, \Kv(1)$, resp. In general, we underline a value to
denote the urgent version of the PAFAS$_s$ process that represents the
corresponding variable. 
\begin{example}\rm\label{ex:pet-1}
Consider the following timed computation from $P$.

\begin{tracex}
$\begin{array}{l c c l l}
\ppath{}{P}{=}{(P_1 \,\|_{\emptyset}\, P_2) \,\|_B\, \vp(\fa,\fa,1)} 
{\nar{\req_1 \; \wtb{1} \;  \wk{2} \; \req_2  \; \wtb{2} \; \wk{1} }{} }
{}\\
\ppath{}{}{}{(P_{11} \,\|_{\emptyset}\, P_{21}) \,\|_B\, \vp(\tr,\tr,1)}{
\nar{\rtb{2} \; \rtb{1}}{} }{}\\
\ppath{}{R}{ = }{(P_{12} \,\|_{\emptyset}\, P_{22}) \,\|_B\,
\vp(\tr,\tr,1)}{ \sar{1} }{}\\
\ppath{}{\underline{R}}{=}{(\underline{P_{12}} \,\|_{\emptyset}\,
\underline{P_{22}}) \,\|_B\, \underline{\vp(\tr,\tr,1)}}{ \nar{\rk{1}}{}
}{}\\
\ppath{}{}{}{(\underline{P_{12}} \,\|_{\emptyset}\, P_{21})
\,\|_B\, \vp(\underline{\tr},\underline{\tr},1)}{ \nar{\rtb{1}}{} }{}\\
\ppath{}{Q}{ = }{(\underline{P_{12}} \,\|_{\emptyset}\, P_{22})
\,\|_B\, \vp(\tr,\underline{\tr},1)} { \sar{1} {}}  {\underline{R}}\\
\end{array}$
\end{tracex}

\noindent Process $R$ can only perform $\rk{1}$ as a synchronisation
between $\Kv(1)$ and either $P_{12}$ or $P_{22}$; after the first  1-step,
this action becomes urgent. Once in $\underline{R}$, we perform $\rk{1}$
and $\underline{\Kv}(1)$ evolves into $\Kv(1)$ which can delay  $\rk{1}$.
As a consequence, $Q$ can refuse to perform $\rk{1}$ and, since this is its
only activated action, $Q \nar{1}{} \underline{R}$. The execution sequence
$\petersona' = P/B \nar{\req_1 \,\tau^2 \, \req_2 \, \tau^4}{} R/B 
\nar{\tau^2 }{} R /B \ldots$ is fair but not live since no process will
ever enter the critical section; $\underline{R}/B \nar{\tau^2 }{} Q/B
\nar{1}{} \underline{R}/B$ corresponds to a catastrophic cycle in the
reduced transition system of $\petersona'_{io}$. 
\end{example}

This example describes a scenario where process $\proc_1$ will never move
because process $\proc_2$ repeatedly reads variables $k$ and $b_1$. There
is another fair run where $\proc_1$, reading variable $b_2$, can
repeatedly delay and, thus,  indefinitely block $\proc_2$ that wants to
write it. On the contrary, the representation of program variables we use
in Def.~\ref{pet-algo1} ensures the liveness of the hardware under the
assumption of fairness of actions; namely, it ensures that no process can
be indefinitely blocked by infinite reading.

\paragraph{Lamport's algorithm}
There are $n \geq 2$ processes and $n$ Boolean-valued variables $b_i$ ($i =
1 \ldots n$), each with initial value false;  only $\proc_i$ writes $b_i$,
but all the processes can read it. The $i$-th process is described below:

\vspace{0.15cm}
\small
\begin{algorithm}{Lamport}{}
{\bf var} \; j:integer;\\
\begin{WHILE}{$true$}
\langle \mbox{non-critical section} \rangle; \\
b_i \= true;\\
\begin{FOR}{j \= 1 \TO i-1}
\begin{IF}{b_j}
b_i \= false; \\
{\bf while} \; b_j \; {\bf do \; skip};\\
{\bf goto} \;  4;
\end{IF}
\end{FOR}\\
\begin{FOR}{j \= i+1 \TO n}
{\bf while} \; b_j \; {\bf do \; skip};
\end{FOR}\\
\langle \mbox{critical section} \rangle; \\
b_i \= false;
\end{WHILE}
\end{algorithm}
\normalsize
Now we provide the PAFAS$_s$ specification in case of $n=2$ processes.
\begin{definition}\label{lamp-algo1}\rm({\em Lamport's algorithm})
Again we first define the family of PAFAS$_s$ processes representing the
program variables. Let $\Bv_i(\fa) = \{\rfb{\it i}, \wfb{\it i}\}
\rop \wtb{\it i}. \Bv_i(\tr)$ and $\Bv_i(\tr) = \{\rtb{\it i}, \wtb{\it
i}\} \rop \wfb{\it i}. \Bv_i(\fa)$ where $i \in \{1,2\}$.
%
%
We also define $\vp(b_1, b_2) = \Bv_1(b_1) \,\|_{\emptyset}\,
\Bv_2(b_2)$ where $b_1, b_2 \in \bool$. 

\noindent Processes $\proc_1$ and $\proc_2$ are represented by:

$\begin{array}{l l}
P_1 = \req_1.\wtb{1}.P_{11} + \tau.P_{1} \quad\quad & \quad\quad 
P_2 = \req_2.\wtb{2}.P_{21} + \tau.P_{2} \\
P_{11} =\rfb{2}.P_{12} + \rtb{2}.P_{11} \quad \quad & \quad \quad
P_{21} = \rfb{1}.P_{23} + \rtb{1}.\wfb{2}.P_{22}\\
P_{12} = \cs_1.\wfb{1}.P_{1}\quad \quad & \quad \quad 
P_{22} = \rfb{1}.\wtb{2}.P_{21} + \rtb{1}.P_{22}\\
& \quad \quad P_{23} = \cs_2.\wfb{2}.P_{2}
\end{array}$

\noindent
Finally $\lamporta = ((P_1 \:\|_\emptyset\: P_2) \:\|_B\:
\vp(\fa,\fa))/B$. 
\end{definition}
Note that now we regard as non-blocking not only the reading of a variable
but also its writing in case that the  written value equals the current
one. This kind of re-write does not change the state of the variable and 
can be thought of as a non-destructive or non-consuming operation (allowing
potential concurrent behaviour). This way of accessing a variable is not
new. It has been implemented e.g.\ in area of database. Unlike in $\peterson$'s specification, we make this assumption on the hardware to show that $\lamport$'s algorithm is not live with respect to $\proc_2$:

\begin{proposition}\label{prop:lamp-live}
If we assume fairness of actions, \lamporta\ is live for process
$\proc_1$ but {\em not} for process $\proc_2$. 
\end{proposition}

\begin{proof}
\lamporta\ is not live for $\proc_2$ because $\lamporta_{io}$ has
catastrophic cycles. To prove the other statement, we need symmetric
changes; namely, we rename actions $\req_1$ and $\cs_1$ into $in$ and $out$
resp. and actions $\req_2$ and $\cs_2$ into $\tau$; we also delete the
$\tau$-summand of process $\proc_1$. Since this modified process does not
have catastrophic cycles, we conclude that \lamporta\ is live for process
$\proc_1$.
\end{proof}

Prop.~\ref{prop:lamp-live} still holds if we use the same representation of
program variables as in Def.~\ref{pet-algo1}, while we lose liveness for
$\proc_1$ whenever processes representing program variables are those used
for $\petersona'$. Then, while reading variable $b_1$, process $\proc_2$
can forever block the other process that wants to write it. The next
example explains why \lamporta\ is not live for process $\proc_2$.
\begin{example}\rm\label{ex:lamp-1}
The following timed computation corresponds to an execution sequence from
$\lamporta = L/B$ which is fair but not live since process $\proc_2$ never
enters its critical section.

\begin{tracex}
$\begin{array}{l c c l l}
\ppath{}{L}{=}{(P_1 \,\|_{\emptyset}\, P_2) \,\|_B\, \vp(\fa,\fa)} 
{\nar{\req_1 \; \req_2  }{} } {}\\
\ppath{}{}{}{(\wtb{1}.P_{11} \,\|_{\emptyset}\, \wtb{2}.P_{21}) \,\|_B\,
\vp(\fa,\fa)} {\nar{\wtb{1} \; \rfb{2} }{} } {}\\
\ppath{}{}{}{(P_{12} \,\|_{\emptyset}\, \wtb{2}.P_{21}) \,\|_B\,
\vp(\tr,\fa)} {\nar{\wtb{2} \; \rtb{1} \; \wfb{2} }{} } {}\\
\ppath{}{R}{=}{(P_{12} \,\|_{\emptyset}\, P_{22}) \,\|_B\, \vp(\tr,\fa)}
{\nar{ 1 }{} } {}\\
\ppath{}{\underline{R}}{=}{(\underline{P_{12}} \,\|_{\emptyset}\,
\underline{P_{22}}) \,\|_B\, \underline{\vp(\tr,\fa)}} {\nar{ \cs_1 \;
\wfb{1} }{} } {}\\
\ppath{}{}{}{(P_{1} \,\|_{\emptyset}\, \underline{P_{22}}) \,\|_B\,
\vp(\fa,\underline{\fa})} {\nar{ \req_1 \; \wtb{1} \; \rfb{2} }{} } {}\\
\ppath{}{Q}{=}{(P_{12} \,\|_{\emptyset}\, \underline{P_{22}}) \,\|_B\,
\vp(\tr,\underline{\fa})} {\nar{ 1 }{} } {} \underline{R}\\
\end{array}$
\end{tracex} 

\noindent $R$ can do either a $\cs_1$- or a $\rtb{1}$-action (due to
a synchronisation between $P_{22}$ and $\Bv_1(\tr)$); both actions become
urgent after the first 1-step. Later, we perform  $\cs_1$ followed by
$\wfb{1}$ (and, hence, $\underline{\Bv_1(\tr)}$ evolves into $\Bv_1(\fa)$)
which, in turn, is followed by $\req_1$ and $\wtb{1}$. At this stage,
$\Bv_1(\fa)$ becomes $\Bv_1(\tr)$ and $Q$ can refuse to perform its
activated actions, again $\cs_1$ and $\rtb{1}$, and evolve into
$\underline{R}$.  Finally, $\underline{R}/B \nar{ \cs_1 \, \tau \, \req_1
\, \tau^2}{} Q/B \nar{1}{} \underline{R}/B$ corresponds to a catastrophic
cycle in the reduced transition system of $\lamporta_{io}$.
A key observation here is that process $\proc_1$, along this cycle,
continuously changes the value of $b_1$ from {\em true} to {\em false} and
vice versa. Consequently, the PAFAS$_s$ process representing this
variable always offers a new instance of $\rfb{1}$ and $\rtb{1}$ to its
synchronisation partners, and in particular to $\underline{P_{22}}$. So,
any possible move of process $\proc_2$ can be arbitrarily delayed (and,
hence, this process can indefinitely be blocked) even in fair traces. No
reasonable assumption about program variables can prevent this unwanted
behaviour under weak fairness.
\end{example}

\vspace{-0.5cm}
\paragraph{Dijkstra's algorithm}
This algorithm considers $n \geq 2$ processes that share two Boolean-valued
arrays $b$ and $c$ (whose components are initialised to true) and a turn
variable $k$ initially chosen in $\{1,2,\ldots, n\}$. The $i$-th process is described below:
\begin{algorithm}{Dijkstra}{}
{\bf var} \; j:integer;\\
\begin{WHILE}{$true$}
\langle \mbox{non-critical section} \rangle; \\
b[i] \= false;\\
\begin{IF}{k \neq i}
c[i] \= true;\\
{\bf if} \; b[k] \; {\bf then} \; k \= i; \\
{\bf goto} \; 5;
\ELSE 
c[i] \= false;\\
\begin{FOR}{j \= 1 \TO n}
{\bf if} \; j \neq i  \mbox{ and } \neg c[j] \;{\bf then \; goto} \; 5;
\end{FOR}
\end{IF}\\
\langle \mbox{critical section} \rangle; \\
c[i] \= true; \\
b[i] \= true;
\end{WHILE}
\end{algorithm}
\noindent Again we provide the PAFAS$_s$ representation in case of $n=2$
processes.
\begin{definition}\rm\label{def:dijkstra1} ({\it Dijkstra's algorithm}) 
Components of the array $b$ are represented by processes $\Bv_i(\fa)$
and $\Bv_i(\tr)$ ($i=1,2$) in Def.~\ref{lamp-algo1}. The other shared
variables are defined similarly.
%
Let $i=1,2$, $b_i, c_i \in \bool$, and $k \in \kvalues$; as usual, 
$\vp(b_1, b_2, c_1, c_2, k)$ denotes the parallel composition of all
program variables. Its definition is as expected and, hence, omitted.
Processes $\proc_1$ and $\proc_2$ are instead given below:

\hspace{1cm}
$\begin{array}{l l}
P_{1} = \req_1.\wfb{1}.P_{11}  + \tau.P_1 & 
P_{2} = \req_2.\wfb{2}.P_{21} + \tau.P_2 \\
P_{11} = \rk{1}.P_{15} + \rk{2}.\wtc{1}.P_{12} & 
P_{21} = \rk{2}.P_{25} + \rk{1}.\wtc{2}.P_{22} \\
P_{12} = get.(\rk{1}.P_{13} + \rk{2}.P_{14}) &
P_{22} = get.(\rk{2}.P_{23} + \rk{1}.P_{24})\\
P_{13} = \rtb{1}.put.\wk{1}.P_{11} + \rfb{1}.put.P_{11}\qquad &
P_{23} = \rtb{2}.put.\wk{2}.P_{21} + \rfb{2}.put.P_{21}\\
P_{14} = \rtb{2}.put.\wk{1}.P_{11} + \rfb{2}.put.P_{11} &
P_{24} = \rtb{1}.put.\wk{2}.P_{21} + \rfb{1}.put.P_{21}\\
P_{15}= \wfc{1}.(\rfc{2}.P_{11} + \rtc{2}.P_{16}) &
P_{25}= \wfc{2}.(\rfc{1}.P_{21} + \rtc{1}.P_{26}) \\
P_{16}= \cs_1.\wtc{1}.\wtb{1}.P_1 &
P_{26}= \cs_2.\wtc{2}.\wtb{2}.P_2 
\end{array}$

\vspace{0.1cm}

\noindent Dijkstra's algorithm is defined as $\dijkstraa =
(((P_1 \:\|_\emptyset\: P_2) \:\|_{\{get, put\}} \BK\: ) \:\|_B\: \vp(\tr,
\tr, \tr, \tr, 1))/(B \cup \{get,put\})$ where $\BK= get.put.\BK$.

As in~\cite{Walker89} we must ensure that whenever, during the execution of
the statement ``{\bf if $b[k]$ then $k \leftarrow i$}'', process $\proc_i$
has read variable $k$ but not yet $b[k]$, the other process cannot change
the value of the former variable. Note that $\BK$ locks the variable $k$ in
writing mode when evaluating $b[k]$. Indeed, after a $get$-action, $k$ can
be written only after a subsequent $put$-action, i.e.\ once $b[k]$ has been
read.
\end{definition}

As other papers in the literature (see e.g.~\cite{Bogunovic03}), we cannot
prove the liveness of the algorithm\footnote{Paper~\cite{Bogunovic03}
studies the liveness of the same algorithms we consider here except for
Lamport. In~\cite{Bogunovic03} it has been proven that Peterson and Knuth
are live, but Dijkstra is not.}. 
In case $k$ is  $1$, process $\proc_1$ can immediately enter its critical
section (after setting $b[1]$ to false, both conditions $k \neq 1$ and
$\neg c[2]$ are false), while process $\proc_2$ has to wait until $b[1]$
becomes true (i.e.\ until $\proc_1$ ends its critical section) and, hence,
it can change $k$. If $\proc_1$ is fast enough to perform its critical
section, reset variables $c[1]$ and $b[1]$, and submit a further request
(again, by setting $b[1]$ to false) before $\proc_2$ can actually read
$b[1]$, the latter process can never enter its critical section. This
scenario is fair and, hence, admissible; see e.g.\ in~\cite{Bogunovic03}
where Dijkstra is analysed by exploiting the model checker SMV
(\cite{Bogunovic03} and references therein). The fairness notion assumed
in~\cite{Bogunovic03} ensures that each process executes infinitely often
and that no process can stay in its critical or non-critical section
forever. The next example shows that the above scenario is also admissible
if one assumes fairness of actions and introduces reasonable non-blocking
behaviours.
\begin{example}\rm\label{ex:dijkstra1}
Let us consider the following timed computation:

\begin{tracex}
$\begin{array}{l l l l l}
\ppath{D}{=}{((P_1 \,\|_\emptyset \, P_2) \,\|_{\{get, put\}}\, \BK)
\,\|_B\, \vp(\tr, \tr, \tr, \tr,1)}{\nar{ \req_1 \; \wfb{1} \; \rk{1}
}{}}{}\\
\ppath{}{}{((P_{15} \,\|_\emptyset\, P_2) \,\|_{\{get, put\}}\, \BK)
\,\|_B\, \vp(\fa, \tr, \tr, \tr,1)}{\nar{ \wfc{1} \; \rtc{2} }{}}{}\\
\ppath{}{}{((P_{16} \,\|_\emptyset\, P_2) \,\|_{\{get, put\}}\, \BK)
\,\|_B\, \vp(\fa, \tr, \fa, \tr,1)}{\nar{ \req_2 \; \wfb{2} }{}}{}\\
\ppath{}{}{((P_{16} \,\|_\emptyset\, P_{21}) \,\|_{\{get, put\}}\, \BK)
\,\|_B\, \vp(\fa, \fa, \fa, \tr,1)}{\nar{ \rk{1} \; \wtc{2} }{}}{}\\
\ppath{}{}{((P_{16} \,\|_\emptyset\, P_{22}) \,\|_{\{get, put\}}\, \BK)
\,\|_B\, \vp(\fa, \fa, \fa, \tr,1)}{\nar{ get \; \rk{1} }{}}{}\\
\ppath{R}{=}{((P_{16} \,\|_\emptyset\, P_{24}) \,\|_{\{get, put\}}\,
put.\BK) \,\|_B\, \vp(\fa, \fa, \fa, \tr,1)}{\nar{ 1 }{}}{}\\
\ppath{\underline{R}}{=}{((\underline{P_{16}} \,\|_\emptyset\,
\underline{P_{24}}) \,\|_{\{get, put\}}\, \underline{put}.\BK) \,\|_B\,
\underline{\vp(\fa, \fa, \fa, \tr,1)}}{\nar{ \cs_1 \; \wtc{1} \; \wtb{1}
}{}}{}\\
\ppath{}{}{((P_{1} \,\|_\emptyset\, \underline{P_{24}}) \,\|_{\{get,
put\}}\, \underline{put}.\BK) \,\|_B\, \vp(\tr, \underline{\fa}, \tr,
\underline{\tr}, \underline{1})}{\nar{ \req_1 \; \wfb{1} \; \rk{1}}{}}{}\\
\ppath{}{}{((P_{15} \,\|_\emptyset\, \underline{P_{24}}) \,\|_{\{get,
put\}}\, \underline{put}.\BK) \,\|_B\, \vp(\fa, \underline{\fa}, \tr,
\underline{\tr}, \underline{1})}{\nar{ \wfc{1} \; \rtc{2} } {}}{}\\
\ppath{Q}{=}{((P_{16} \,\|_\emptyset\, \underline{P_{24}}) \,\|_{\{get,
put\}}\, \underline{put}.\BK) \,\|_B\, \vp(\fa, \underline{\fa}, \fa,
\underline{\tr}, \underline{1})}{\nar{ 1 } {}}{ \underline{R}}\\
\end{array}$
\end{tracex}

Along the cycle $\underline{R}/B \nar{ \cs_1 \, \tau^2 \, \req_1 \,
\tau^4}{} Q/B \nar{1}{} \underline{R}/B$, the process $\proc_1$ repeatedly
changes the value of $b_1$  from {\em false} to {\em true} and vice versa.
As in Example~\ref{ex:lamp-1}, this means that it can block forever process
$\proc_2$.
\end{example}

\begin{proposition}\label{prop:dijkstra} 
\dijkstraa\ is not live under the assumption of fairness of actions.
\end{proposition}

\paragraph{Knuth's algorithm}
There are two processes $\proc_1$ and $\proc_2$, two variables $c_1$ and
$c_2$ that take values in $\{0, 1, 2\}$ and whose initial value is 0, and a
turn variable $k$ that takes values in $\{1, 2\}$
and whose initial value is arbitrary.  Process $P_i$ ($i=1,2$) is
described as follows, where $j$ is the index of the other process:
\begin{algorithm}{Knuth}{}
\begin{WHILE}{$true$}
\langle \mbox{non-critical section} \rangle; \\
c_i \= 1;\\
{\bf if} \; k = i  \; {\bf then \; goto} \; 6;\\
\; {\bf if} \; c_j  \neq 0  \; {\bf then \; goto} \; 4;\\
c_i \= 2;\\
{\bf if} \; c_j =2  \; {\bf then \; goto} \; 3;\\
k \= i;\\
\langle \mbox{critical section} \rangle; \\
k \= j; \\
c_i \= 0;
\end{WHILE}
\end{algorithm}
\begin{definition}\rm\label{def:knuth1} ({\it $\knuth$'s algorithm}) 
The turn variable $k$ is given in Def.~\ref{pet-algo1} and modelled according to Def.~\ref{def:dijkstra1}. Variables $c_1$
and $c_2$ are represented as follows, where $i = 1, 2$:

\vspace{0.1cm}
\hspace{3.5cm}
$\begin{array}{l c l}
\Cv_i(0) & = & \{\rc{i}0, \wc{i}0\} \rop (\wc{i}1.\Cv_i(1) +
\wc{i}2.\Cv_i(2)) \\
\Cv_i(1) & = & \{\rc{i}1, \wc{i}1\} \rop (\wc{i}0.\Cv_i(0) +
\wc{i}2.\Cv_i(2))\\
\Cv_i(2) & = & \{\rc{i}2, \wc{i}2\} \rop (\wc{i}0.\Cv_i(0) +
\wc{i}1.\Cv_i(1))
\end{array}$

\vspace{0.15cm}

\noindent Let  $c_1, c_2 \in \{0,1,2\}$  and $k \in \kvalues$. We
let $\vp( c_1, c_2, k)$ to be the parallel composition of all program
variables. Moreover, processes $\proc_1$ and $\proc_2$ are defined as
follows: 

\vspace{0.1cm}

\hspace{1.5cm}
$\begin{array}{l l}
P_{1} = \req_1.\wc{1}1.P_{11} + \tau.P_{1} \quad\quad &
P_{2} = \req_2.\wc{2}1.P_{21} + \tau.P_{2} \\
P_{11} = \rk{1}.P_{13} + \rk{2}.P_{12} &
P_{21} = \rk{2}.P_{23} + \rk{1}.P_{22} \\
P_{12} = \rc{2}0.P_{13} + \rc{2}1.P_{11} + \rc{2}2.P_{11} &
P_{22} = \rc{1}0.P_{23} + \rc{1}1.P_{21} + \rc{1}2.P_{21} \\
P_{13} = \wc{1}2.P_{14} &
P_{23} = \wc{2}2.P_{24} \\
\end{array}$

\hspace{1.5cm}
$\begin{array}{l l}
P_{14} = \rc{2}0.P_{15} + \rc{2}1.P_{15} + \rc{2}2.P_{16} &
P_{24} = \rc{1}0.P_{25} + \rc{1}1.P_{25} + \rc{1}2.P_{26} \\
P_{15} = \wk{1}.\cs_1.\wk{2}.\wc{1}0.P_1 &
P_{25} = \wk{2}.\cs_2.\wk{1}.\wc{2}0.P_2 \\
P_{16} = \wc{1}1.P_{11} &
P_{26} = \wc{2}1.P_{21} 
\end{array}$

\noindent We define $\knuthb = ((P_1 \:\|_\emptyset \: P_2) \:\|_B\: \vp(0,
0, 1))/B$. 
\end{definition}
We now provide an example that shows the existence of a catastrophic cycle
in the reduced transition system of the modified \knutha. This example also
implies Prop.~\ref{prop:knuth}.

\begin{example}\rm\label{ex:knuth}
Let us consider the following timed computation:

$\begin{array}{l l l l l}
\ppath{K}{=}{(P_1 \,\|_\emptyset\, P_2)  \,\|_B\, \vp(0,0,1)}{%
\nar{ \req_2 \, \wc{2}1 }{}}{}\\
\ppath{}{}{(P_1 \,\|_\emptyset\, P_{21})  \,\|_B\, \vp(0,1,1)}{%
\nar{ \rk{1} \, \rc{1}0 \, \wc{2}2 }{}}{}\\
\ppath{}{}{(P_1 \,\|_\emptyset\, P_{24})  \,\|_B\, \vp(0,2,1)}{%
\nar{ \req_1 \, \wc{1}1 }{}}{}\\
\ppath{}{}{(P_{11} \,\|_\emptyset\, P_{24})  \,\|_B\, \vp(1,2,1)}{%
\nar{ \rk{1} \, \wc{1}2 }{}}{}\\

\ppath{R}{=}{(P_{14} \,\|_\emptyset\, P_{24})  \,\|_B\, \vp(2,2,1)}{%
\nar{ 1 }{}}{}\\
\ppath{\underline{R}}{=}{(\underline{P_{14}} \,\|_\emptyset\,
\underline{P_{24}}) 
\,\|_B\, \underline{\vp(2,2,1)}}{\nar{ \rc{2}2 }{}}{}\\
\ppath{}{}{(P_{16} \,\|_\emptyset\, \underline{P_{24}}) 
\,\|_B\, \underline{\vp(2,2,1)}}{\nar{ \wc{1}1 \; \rk{1} \; \wc{1}2
}{}}{}\\
%
\ppath{Q}{=}{(P_{14} \,\|_\emptyset\, \underline{P_{24}}) 
\,\|_B\, \vp(2,\underline{2},\underline{1})}{\nar{ 1 }{}}{\underline{R}} 
\end{array}$

\vspace{0.15cm}

\noindent Once in $R$, process $\proc_1$ cannot enter its critical section
because $c_2$ is 2; but, the value of this variable will
never change because $\proc_2$ is blocked. Moreover, as in
Examples~\ref{ex:lamp-1} and~\ref{ex:dijkstra1}, repeated changes of
variable $c_1$ (from 2 to 1 and vice versa) allows a further 1-step in
$Q$.
The execution sequence $\knutha = K/B \nar{\req_2 \, \tau^4 \, \req_2 \,
\tau^3}{} R/B  \nar{ \tau^4 }{} R /B \ldots$ is fair but not live since
process $\proc_2$ never enters its critical section. Let us finally notice
that Knuth is live e.g.\ in~\cite{Bogunovic03} since the above execution
sequence is not fair as defined there, and hence not admissible, because
process $\proc_2$ does not execute infinitely often. 
\end{example}
\begin{proposition}\label{prop:knuth}
\knutha\ is not live under the assumption of fairness of actions.
\end{proposition}
\section{Related works and Conclusion} \label{sec:walker}
This work partly originates from~\cite{Walker89} where Walker aimed at
verifying six MUTEX algorithms with the Concurrency
Workbench~\cite{CleavelandPS89} (CWB, for short). 
Walker translated the algorithms into CCS and then verified the safety
property that the two competing processes are never in their critical
sections at the same time. Regarding the liveness property, Walker first
considered the following interpretation -- which could be expressed as a
modal mu-calculus formula and then checked with the CWB:
\begin{center}
\begin{minipage}{15.5cm}
\small
\em
An algorithm is live if whenever at some point in a computation the
process $\proc_i$ requests the execution of its critical section, then in
any continuation from that point in which between them the processes
execute an infinite number of critical sections, $\proc_i$ performs its
critical section at least once. 
\end{minipage}
\end{center}
The fairness (or progress) assumption assumed here is that infinitely often
a critical section is entered. This assumption allows a run where one
process enters its critical section repeatedly, while the other one
requests the execution of its critical section, but then -- for no good
reason at all -- refuses to take the necessary steps to actually enter it.
So, it may be no surprising that four of the six algorithms (Dekker,
Dijkstra, Lamport and Hyman) fail to satisfy this property. Moreover, in
order to economize on computational effort, the six algorithms
in~\cite{Walker89} have been minimized w.r.t. weak bisimulation. This
allowed Walker to ignore some $\tau$-loops that could invalidate the
liveness property. And, indeed, all of them are not live whenever
the formula expressing the first interpretation of liveness is evaluated
over the transition system that does not abstract from $\tau$'s. By
examining process $\proc_i$, it is clear that these $\tau$-loops arise,
e.g. in Peterson, from repeated reading and writing of variables by the
same process. This is common to all the algorithms and it is not introduced
by the translation into CCS (or in PAFAS). Rather its presence reflects the
faithfulness of the translation itself.

Then, Walker considered the same liveness property we study in
Section~\ref{algos}. To establish that any of the algorithms is live under
this second interpretation, Walker added some assumption. Indeed,
one characteristic of the $\tau$-loops arising from repeated reading and
writing of variables by one process is that the other one is excluded from
an infinite computation of the system. It is natural to ask if {\em
only} the presence of such `unfair' loops prevents any of the
algorithms from being live. So, Walker proposed to use enriched formulas
of the form $F \Rightarrow P$ where $P$ is the property of interest (i.e.\
liveness) and $F$ is a fairness assumption that assumes as admissible
only those paths to which each process contributes infinitely often. 
Even if at the time of writing no automated analysis was possible, Walker
discussed how fairness could be assumed. The basic idea  is to tag
each action with a unique {\em probe} or label; then, we can say the $i$-th
process $\proc_i$ contributes infinitely often to a computation whenever
none of its probes is continuously possible from a certain point on.
Finally, the liveness under this fairness
assumption is expressed by letting $K_i$ be the set of all probes of
$\proc_i$ and defining  $FairLive = FairLive_1 \wedge FairLive_2$
where 
$FairLive_i\ =\ (\bigwedge_{a \in K_{i}}\ GF[a] false) \Rightarrow
G(\langle\req_i\rangle true  \Rightarrow F\langle\cs_i\rangle true)$,
and the operators $G$ (always), $F$ (future), $\langle\rangle$ (possibly)
and $[]$ (necessarily) are standard modal logics operators.

\begin{table}
\centering
\small
\begin{tabular}{l c c l c c l c c}
 & CWBNC & \fase & &  CWBNC & \fase\  & &  CWBNC & \fase\  \\ \hline
\midrule
$\dekker$ & \checkn\ & \checky\ & \peterson\ & \checkn\ & \checky\
& $\knuth$ & \checkn\ & \checkn\ \\
$\dijkstra$ & \checkn\ & \checkn\ & \lamport\ & \checkn\ & \checkn\ \\
\midrule
\end{tabular}
\caption{Liveness of MUTEX solutions: CWBNC vs. \fase.}
\label{tab:comparing2}
\end{table}

This fairness induced with probes is closely related to fairness of actions
as it has been defined in~\cite{CostaS84,CostaS87}. W.r.t. our
characterisation (cf.\ Section~\ref{PAFAS-s}),  the  main difference is
that, instead of time and time passing, probes are used to decide whenever
an action is continuously enabled along a computation and, hence, must be
performed eventually. To allow a comparison, we have implemented these
ideas within the {\em Concurrency Workbench of the New
Century}~\cite{CLS00} (CWBNC, for short) that, unlike CWB, can handle modal
formulas with fairness constrains. To be able to attach a probe to each
process action, the algorithms have been translated into {\em Timed CCS}
(this is not possible by using the standard CCS language); probes are
introduced by annotating synchronisation actions or $\tau$'s. For instance,
the $i$-th processes of Peterson can be defined as follows:

\hspace{2.5cm}
$\begin{array}{l l}
P_i= \wtb{\it i}(req_i).\wk{j}(a_i). P_{i1} &
P_{i2}= \rk{j}(a_{i}).P_{i1} + \rk{i}(a_{i}).P_{i3}\\
P_{i1}= \rfb{\it j}(b_{i}). P_{i3} + \rtb{\it j}(b_{i}).P_{i2} \quad &
P_{i3}=  \cs_i(cs_i).\wfb{\it i}(a_{i}).P_{i}
\end{array}$

Note that two consecutive actions (as, e.g., $\wtb{\it i}$ and $\wk{j}$ in
$P_i$) never have the same label. Moreover, since the overall number of
labels impacts on the computational effort (see below), we also try to
reduce the number of labels we use.  For example, we can reuse $a_i$ to
label the actions of $P_{i2}$ because none of its actions is adiacent to
$\wk{j}$ and this action has already been executed once $P_{i2}$ is
reached. 

Whenever an action is performed, the corresponding label becomes {\it
visible}\footnote{E.g.,\ if $P_i$ synchronises with $\Bv_i(\fa)$ on the
execution of $\wtb{\it i}$, the label $@req_i$ becomes visible; similarly,
whenever process $P_i$ executes $\cs_i$ we get the label $@cs_i$.} and can
be used as a probe in $FairLive$. 
Table~\ref{tab:comparing2} shows that all the algorithms we consider are
not live according to this second liveness interpretation (also in this
setting, Lamport is live for process $\proc_1$ but not for $\proc_2$). As
an example, consider a path from Peterson along
which the first process reaches $P_{11} = \rfb{2}(b_{1}). P_{13} +
\rtb{2}(b_{1}).P_{12}$, $b_2$ is true and $k$ is 2. Once in such a state,
process $\proc_1$ can read $b_2$ and $k$ and come back to $P_{11}$. Along
this cycle, no probe of $\proc_1$ is continuously possible (probes
$b_1$ and $a_1$ are alternately possible) but $\cs_1$ will never be
performed. So, $FairLive_1$ is false and Peterson is not live. As in
Example~\ref{ex:pet-1}, there is a path along which a
process can be indefinitely blocked by repeated reading. Also in this
setting, the liveness of the algorithm strongly depends on the liveness of
the hardware, i.e.\ on the the possibility of making some behaviours
non-blocking.

As a further counter-check, we again consider Peterson but now we tag its
actions in such a way that the same probe is associated to all the  actions
that appear along consecutive reading (trying to simulate the intuition
behind non-blocking behaviours). So, let us replace $P_{i2}$ with 
$P_{i2}= \rk{j}(b_{i}).P_{i1} + \rk{i}(b_{i}).P_{i3}$. 
Now, whenever in $P_{11}$ and assuming $b_2$ and $k$ equal to true and 2,
the process $\proc_1$ can still repeatedly read variables $b_2$ and
$k$, but the corresponding path is not fair because probe $b_1$ is
continuously possible. With these probes, Peterson and Dekker turn out to
be live.
So, probes can be used to somehow simulate non-blocking actions. But they
must be added and (whenever necessary) tuned by the user by hand. This
task is subject to errors and wrong assumptions that would give erroneous
results.  On the contrary, \fase\ can be more easily used by also a
non-expert user that has only to decide whether (and, in case,
which) non-blocking behaviours are necessary. In our opinion, the use of
probes requires a deeper knowledge of the problem and much
more attention in both modelling and analysis phases. 

Another difference between the two approaches deals with {\em performance}
issues. In Table \ref{tab:comparing3} we report the execution time of both $\fase$ and $\cwbnc$ to perform the analysis on the algorithms discussed in this paper. 
In particular, in~\cite{BCCDV09} an efficient algorithm for detecting catastrophic
cycles has been proposed and implemented. This works in time $O(N + E)$
where $N$ and $E$ are, resp., the number of nodes and edges of the
state space of the process. On the contrary, CWBNC uses an on-the-fly model
checking algorithm whose complexity is exponential in the size of the
formula (see~\cite{CLS00}); in our case, this size strongly depends
on the number of probes.

\begin{table}
\centering
\small
\begin{tabular}{l c c l c c l c c}
 & CWBNC & \fase & &  CWBNC & \fase\  & &  CWBNC & \fase\  \\ \hline
\midrule
$\dekker$ & $103125$  & $119$ & \peterson\ & $4844$\ & $34$\
& $\knuth$ & $110391$\ & $166$\ \\
$\dijkstra$ & $110797$\ & $647$\ & \lamport\ & $1734$\ & $22$\ \\
\midrule
\end{tabular}
\caption{Execution time (expressed in milliseconds): CWBNC vs. $\fase$}
\label{tab:comparing3}
\end{table}

\fase\ is a good first step towards the creation of an integrated
framework for the analysis of concurrent systems. The improvements
introduced by the tool (and, in particular, the possibility to easily
check non-functional properties such as liveness) allows us to derive
results -- as those in this paper -- very hard to prove by hand. Since
these results are very promising, we are currently planning to extend
\fase\ in order to improve the analysis of more complex systems with a
larger state space.

\bibliographystyle{eptcs}
\bibliography{mutex-algs}
\end{document}